\documentclass[11pt]{article}
\usepackage{fullpage}
\usepackage{amsmath}
\usepackage{amssymb}
\usepackage{amsthm}
\usepackage{verbatim}
\usepackage{graphicx}
\usepackage{epsfig}
\usepackage{algorithmic}
\usepackage{algorithm}
\usepackage[lofdepth,lotdepth]{subfig}

\usepackage{adjustbox}
\usepackage{enumerate,amsmath,amssymb,chngpage,enumitem,multirow,tabu}

\usepackage{graphicx}
\usepackage{balance}
\usepackage{setspace}
\usepackage{enumerate}
\usepackage{url}
\usepackage{multirow}
\usepackage{cite}
\usepackage{cases}

\newtheorem{claim}{Claim}

\newtheorem{definition}{Definition}
\newtheorem{proposition}{Proposition}
\newtheorem{corollary}{Corollary}

\newtheorem{theorem}{Theorem}

\newtheorem{lemma}{Lemma}

\newcommand{\bfP}{{P}}
\newcommand{\bfA}{{A}}
\newcommand{\bfB}{{B}}
\newcommand{\bfC}{{C}}

\newcommand{\bfH}{{H}}
\newcommand{\bfL}{{L}}
\newcommand{\bfX}{{X}}
\newcommand{\bfY}{{Y}}
\newcommand{\bfS}{{S}}

\newcommand{\rank}{\mbox{rank}}
\newcommand{\sign}{\mbox{sign}}

\newcommand{\tree}{\mathcal{T}}
\newcommand{\child}{\Gamma}
\newcommand{\subgames}{\mathcal{S}}
\newcommand{\cross}{{C}}
\newcommand{\osp}{\mathcal{O}}
\newcommand{\laminar}{\mathcal{R}}

\begin{document}

\title{The Empirical Implications of Rank in Bimatrix Games}
\author{
Siddharth Barman\thanks{Center for the Mathematics of Information, California Inst. of Technology. \tt{barman@caltech.edu}.}
\and
Umang Bhaskar\thanks{Center for the Mathematics of Information, California Inst. of Technology. \tt{umang@caltech.edu}. }
\and
Federico Echenique\thanks{Humanities and Social Sciences, California Inst. of Technology. \tt{fede@hss.caltech.edu}.}
\and
Adam Wierman\thanks{Computing and Mathematical Sciences, California Inst. of Technology. \tt{adamw@caltech.edu}.}}

\maketitle

\thispagestyle{empty}

\begin{abstract}
We study the structural complexity of bimatrix
games, formalized via {\em rank}, from an empirical perspective.  We
consider a setting where we have data on player behavior in diverse strategic
situations, but where we do not observe the relevant payoff functions. We prove that high
complexity (high rank) has empirical consequences when arbitrary data is considered.
Additionally, we prove that, in more restrictive classes of data (termed
laminar), any observation
is rationalizable using a low-rank game: specifically a zero-sum
game. Hence complexity as a structural property of a game is not always
testable. Finally, we prove a general result connecting the structure
of the feasible data sets with the highest rank that may be needed to
rationalize a set of observations.
\end{abstract}

%\begin{bottomstuff}
%This research was supported by NSF grants CNS-0846025 and CCF-1101470.
%\end{bottomstuff}

\newpage

\setcounter{page}{1}

\section{Introduction}
\label{sec:intro}

This paper studies the ``structural'' complexity of bimatrix games from an \emph{empirical perspective}.  That is, we assume that we have data on players' behavior (choices) within a game-theoretic environment, but are ignorant of their payoffs; and we seek to understand whether the data can be explained via games with ``simple'' structure or whether the data implies that the games have ``complex'' structure, where we take the \emph{rank} of a game as a measure of its structural complexity.  Thus, we seek to characterize the empirical/testable implications of rank in bimatrix games. Recall that the rank of a bimatrix game is the rank of the sum of the payoff matrices of the two players.

While questions related to the complexity of economic models have been
a driving force behind research at the intersection of computer
science and economics, the empirical approach of the current paper is
nonstandard for this literature.  In particular, the dominant
perspective of work in this direction has been an
\emph{algorithmic perspective}. Most work has taken the
economic model to be fixed and literal, and then proceeded to ask
about the computational demands placed on the agents by the model.
For example, in the case of noncooperative games, when the payoffs are
fixed and given, computing a mixed Nash
equilibrium has been shown to be hard~\cite{daskalakis09,etessami10},
even for 2-player
games~\cite{chen09}.\footnote{Computing $3$-player Nash
  Equilibria is FIXP-complete \cite{etessami10}; the $2$-player
  case is PPAD-complete \cite{chen09}. For additional background, we refer the interested reader to \cite{Daskalakis09survey}
 and to \cite{roughgarden10}} In contrast, this work takes an
empirical perspective  motivated by \emph{revealed-preference theory}.

\subsection*{The revealed-preference approach}

Revealed-preference theory seeks to understand the empirical implications of
economic models: given data from an observed phenomenon, the task is
to understand how generally a model is applicable
(e.g., how large the class of explainable data is) and to determine
what instance of the theory is consistent with the data (e.g., to
determine the payoff matrices that are consistent with the
data). The revealed-preference approach has a long tradition in economics,
e.g., see
\cite{samuelson1938note,samuelson1938empirical,houthakker1950revealed,afria67,varia82,varia84}. The
paper \cite{varian2006revealed} contains an excellent survey.

If one thinks of economics, and game theory, as a positive (predictive)
science (and, arguably, the vast majority of economists do), then the
revealed-preference approach is unavoidable.
Economists use game theory to understand the behavior of
human players, and of organizations run by human agents. When
observing such behavior, economists do not have access to
agents' payoffs. Thus, 
payoffs, and utility functions, are really unobservable  theoretical
constructs. Payoffs and utilities do not have an independent empirical
meaning. Instead, one must understand which observable
behaviors have corresponding payoffs such that the
theory predicts the observed behavior.

\subsection*{Noncooperative games}

Our work focuses on revealed-preference theory in the context of
noncooperative games. A data set is a collection of observed choices
taken by a set of agents in different strategic situations. We want to
know when there are payoffs for these agents that can explain the data
i.e., payoffs such that the observations are consistent with  the
theory of Nash equilibrium. More specifically, our goal is to
understand the empirical consequences of the structural complexity of
bimatrix games.  We adopt \emph{rank} as our notion of structural
 complexity. 

We adopt the model of observed data from \cite{sprum00} (see also \cite{yanovskaya1980}). We restrict ourselves to two-player games, and assume a grand set of strategies for each player. Each player is then restricted to a subset of its grand set, and we observe the joint strategy played in this subgame. This is called the observed choice. A data set is a collection of subgames, together with the observed choice in each subgame.

This model of observed data is standard in the revealed-preference
literature. It is a straightforward extension of the classical
single-player setting studied by
\cite{samuelson1938note,houthakker1950revealed}, where the agent is
presented with different sets of alternatives and we observe the
alternative chosen in each case. Our model is also motivated by
experimental economics. For example, \cite{camerer2003behavioral}
describes experiments where treatments differ in the set of strategies
available to the players. A typical example is the difference between
the ultimatum and dictatorship games; another example is
\cite{costa2006cognition}, who look at guessing games and vary the set
of possible guesses subjects can make.

A data set is \emph{rationalizable} if there exist payoff matrices such that the observed choices are strict and pure Nash equilibrium in the corresponding subgames. We adopt strict Nash as a natural discipline for avoiding trivial rationalizations. Without such a discipline, all observed choices can be rationalized by setting all payoff values for the players to be zero. We focus on pure Nash following the literature in economics on revealed-preference theory~\cite{galambos,lee10,sprum00}. There are other solution concepts where revealed-preference theory is relevant, and we discuss some of these in the conclusion.

The rank of the sum of the payoff matrices is our measure of the structural complexity of the rationalization. A paradigmatic notion of a simple game is the class of {\em rank-zero} games, which are simply zero-sum games. Zero-sum games are clearly simple: they are easy to analyze using linear programming; their equilibria are easy to find; and player behavior has an intuitively simple heuristic: min-max. More generally, games with small rank are also intuitively simple.  For example, rank-one games have a particularly nice interpretation: these correspond to games where the welfare generated is simply the product of the marginal welfare generated by the actions of the players.  Further, games with low rank are structurally simple since a matrix of rank-$k$ can be decomposed into $k$ rank-one matrices. Our measure thus captures this intuitive notion of simplicity in bimatrix games. The rank of a game also has implications for the computation of \emph{mixed} Nash equilibria. Mixed Nash equilibria in zero-sum games and rank-one games can be found in polynomial time~\cite{Nisan2007,Adsul2011}, while in rank-$k$ games, for fixed $k$, $\epsilon$-approximate Nash equilibria can be found in time poly($1/\epsilon$)~\cite{kannan2010games}. Hence, rank is an important concept for varied notions of complexity.

\subsection*{Contributions of this paper}

Said succinctly, the goal of this paper is to characterize the
empirical implications of rank in bimatrix games.  That
is, to understand whether data can always be explained via
low-rank (simple) games or whether high-rank (complex) games are
necessary. It is a priori possible that, while highly complex games are
relevant theoretically, they are not needed empirically.

We prove four main results that shed light on the empirical implications of rank in bimatrix games. We first establish
that there exist data sets that
cannot be rationalized using low-rank games (Theorem \ref{thm:highrank}).  Specifically, a data set over $n$ strategies may require a rationalizing game to have rank $\Omega(\sqrt{n})$.  Consequently, the notion of structural complexity
studied here has empirical consequences: it is testable, i.e, refutable via data. This
conclusion is in contrast with other models in economics:
\cite{echengolovinwierman2011}  prove that (computational) complexity
has no empirical consequences for the model of individual consumer
choice. Our result highlights that their conclusion does not hold for the standard noncooperative
model, at least when we interpret complexity as low rank.

In order to refine the conclusion of the first result, our second and third results focus on two cases in which the structure of the observed data is restricted. In our second result, Theorem \ref{thm:rankone}, we prove that if all observations have the players facing the same set of
strategies, then the data can be rationalizable as a simple game, specifically a rank-one game.
Similarly, in our third result, Theorem \ref{thm:laminar}, we again consider a restricted class of data sets in which each set of alternatives has a single observation and different sets are either disjoint or nested, i.e., they have a {\em laminar} structure.  We prove that, in this case, the data is always rationalizable using a zero-sum game. Thus, for this (restrictive) family of data sets, structural complexity (rank) again lacks empirical bite: any data set can be rationalized using the simplest family of games, zero-sum games.

Note that, though restrictive, the classes of data sets studied in the results discussed above are relevant for typical designs used in experimental economics.  In particular, these two classes of data sets correspond to situations where an experimenter either keeps the same actions available during repeated plays or successively adds feasible actions for players.

Our final result shows that for arbitrary data sets, the
rank needed to rationalize the data is tied to the degree of nonlaminarity of the sets of alternatives. This generalizes the result for the laminar case discussed above, which states that laminar data sets can always be rationalized with a zero-rank game. We prove that if the
crossing span (see Definition \ref{defn:crossing-span}) is $k$, then the rank needed to rationalize the
observations is at most $k$. Thus, a data set can refute the assumption of low rank only if the data is rich enough, i.e., has a large crossing span.

Our results give new insights into the empirical consequences of low- and high-rank games. The notion of a high-rank, complex, game does have empirically testable implications, but these implications depend on being able to observe rich families of data. If data is restricted, for example contains a laminar family of subgames, then it is not possible to refute the assumption of low rank. In general, we tie the empirical consequences of rank to the crossing span of the data set.

\subsection*{Relationship to prior work}

The revealed-preference approach is well studied in the
classical context of consumer choice theory (see the large literature
started by \cite{samuelson1938note,houthakker1950revealed,afria67},
and the discussion of revealed preference and complexity in
\cite{echenique2011complexity}, based on
\cite{echengolovinwierman2011}).

There is much less work in the context of noncooperative games. The framework we use in this paper is adopted from \cite{sprum00}, which first studied revealed-preference theory in this context.  Though we use the same setting as \cite{sprum00}, our work differs in an important sense. Sprumont requires that one observes choices from all possible subsets of strategies. In particular, he needs choices from subsets where all players but one are restricted to a single feasible alternative. From such observations, Sprumont infers a single-agent's complete revealed-preference graph, and then construct each agent's payoffs. Our work relaxes this assumption. We consider scenarios where not all possible sets of strategies are presented to the agents. The resulting revealed-preference graph will not be perfectly informative about agents' payoffs.

More recently, there are two other papers that study
noncooperative games via the revealed-preference approach, \cite{galambos,lee10}.
Like us,  \cite{galambos} does not require observed choices in all possible subgames. In \cite{galambos}, the problem of
deciding if a data set is rationalizable by a game where the observations are the only Nash equilibria is proven to be NP-hard. Our results imply that this hardness only arises when the sets of strategies are non-laminar. 
More related to our work is \cite{lee10}, which also studies zero-sum games from the revealed-preference perspective. However, \cite{lee10} has the same restrictions as \cite{sprum00}, i.e., all
possible subsets of strategies are observed.  In this setup, \cite{lee10} proves that zero-sum has a very particular empirical implication (on
top of the ones imposed by Nash rationalizability): the well-known property of exchangeability of Nash equilibria of zero-sum games. Our results differ from those in \cite{lee10} in that we do not require all subsets of strategies to be observable, and that we study low-rank games that are not necessarily zero-sum.
 % Introduction
\section{Preliminaries}
\label{sec:prelims}

\subsection*{Two player normal form games}

In this paper we study two-player games in normal form. A two-player game in normal form is given by a pair of matrices $(A, B)$ of size $n \times n$, which are termed the \emph{payoff matrices} for the players. The first player, also called the row player, has payoff matrix $A$, and the second player, or the column player, has payoff matrix $B$. The strategy set for each player is $[n]$. In this standard setup, if the row player plays strategy $i$ and column player plays strategy $j$, then the payoffs of the two players are $A_{ij}$ and $B_{ij}$ respectively. A \emph{pure Nash equilibrium} is a strategy profile of the players $(i,j) \in [n] \times [n]$ such that for all $i', j' \in [n]$, $A_{ij} \geq A_{i' j}$ and $B_{ij}  \geq B_{i j'}.$ If these inequalities are strict, then $(i,j)$ is said to be a \emph{strict Nash equilibrium}.

Our focus is on games with low rank, where the \emph{rank} of a game is the rank of the matrix $C := A+B$.  For a zero-sum game, $C = \mathbf{0}$.

Additionally, we often consider subgames of a game, which
correspond to restrictions on the strategies available to the
players. A subgame is denoted by $(I, J)$ where $I, J \subseteq
[n]$ and $I$ and $J$ are the strategies available to the row and
column player respectively.

\subsection*{Observable data and the revealed preference question}

The core of the revealed preference approach is the definition of the data that is observable.  In this paper, we adopt the treatment of \cite{sprum00} and we observe the strategies chosen by the two players during a sequence of subgames. No other information about the payoff matrices is observed.

Specifically, a \emph{data set} consists of a set of triples $T = \{(
(i,j), I,J) \}_{i,j,I,J}$ where $(I,J)$ is a subgame, and $i \in I$ and $j \in
J$. Each such triple is called an \emph{observation}. A triple $(
(i,j), I, J )$ denotes that in the subgame $(I,J)$, the row player
picked strategy $i$ and column player picked  $j$. The strategy profile $(i,j)$ is called the observed choice in the observation. 

Given this data, the revealed preference question is to determine for which data sets $T$ there exist payoff matrices $A$ and $B$ such that $(i,j)$ is a pure and \emph{strict} Nash equilibrium in the subgame $(I,J)$ for all triples $( (i,j), I,J)\in T$.  More formally, we
say that a data set $T$ is \emph{rationalizable} if there exist such payoff matrices. Without the discipline of strictness in the definition of rationalizability, data could always be rationalized via trivial games.

\begin{figure}[t]
\begin{center}
\subfloat[Data set that is rationalizable (via a rank one game).]{
\hspace{.49in}
\includegraphics[scale=0.25]{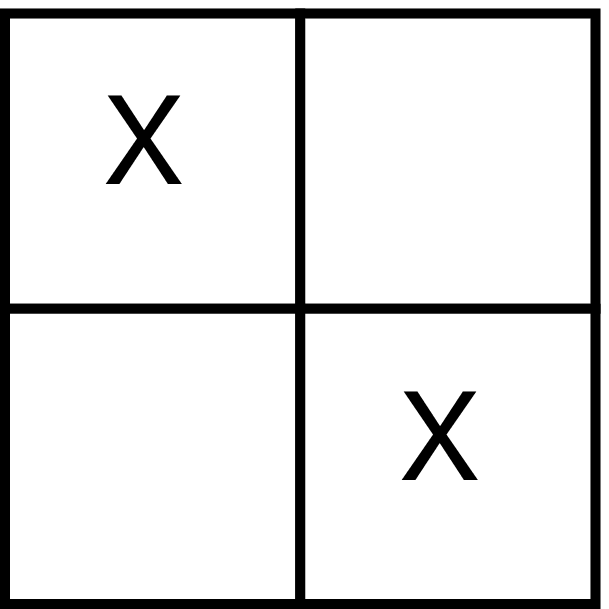} %\hspace{-150pt}
\hspace{.54in}
\label{fig:rat}} \qquad
\subfloat[Data set that is not rationalizable.]{
\hspace{.38in}
\includegraphics[scale=0.25]{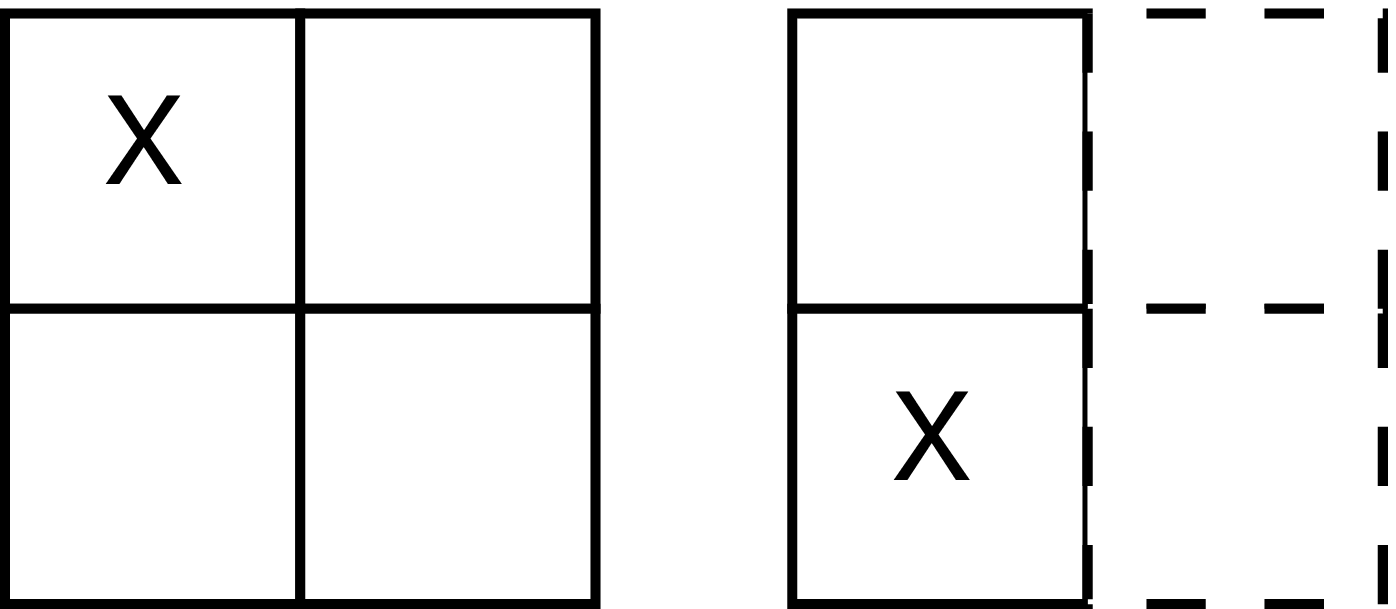}
\hspace{.38in}
\label{fig:nonrat}}
\caption{\emph{Examples of a rationalizable data set and a data set that is not rationalizable.}}
\label{fig:ratnotrat}
\end{center}
\end{figure}

It is important to remark at this point that not all data sets are rationalizable. For example Figure~\ref{fig:ratnotrat} depicts a data set that is rationalizable, and another that is not. Figure~\ref{fig:rat} depicts the data set $T = \{((1,1),\{1,2\},\{1,2\}), ((2,2),\{1,2\},\{1,2\})\}$, rationalizable via
\begin{equation}
A = \left[ \begin{array}{ll}
	2 & 7 \\
	1 & 8
	\end{array} \right] \text{ and }
B = \left[ \begin{array}{ll}
	2 & 1 \\
	7 & 8
	\end{array} \right] \, . \label{e.rankone}
\end{equation}
It is easy to verify that the equilibria are exactly those in the given data set.

In contrast, for a simple example of a game that is not rationalizable consider the data set $T' = \{((1,1),[2],[2]), ((2,1),[2],\{2\})\}$ shown in Figure~\ref{fig:nonrat}. Since all equilibria must be strict, by the first observation, the entry corresponding to the observed equilibria must dominate the entry above it in the row player's payoff matrix. By the second observation, the exact opposite must also be true, which is impossible.

In this paper our goal is not to characterize data sets that are rationalizable. Rather, our goal is to characterize data sets that are rationalizable with \emph{low-rank} bimatrix games, where we use rank as a notion of the structural complexity of the game.

To illustrate the empirical implications of rank briefly, consider again the rationalizable data set in Figure \ref{fig:rat} and the rationalizing payoff matrices in \eqref{e.rankone}.  It is easy to verify that the matrix $C$ has rank one. In fact, this is the lowest rank rationalization possible for this data set. To see that no rank-zero game rationalizes this data set, observe that since the diagonal entries are strict equilibria, each off-diagonal payoff in $A$ must be dominated by the diagonal payoff in the same column, and each off-diagonal payoff in $B$ must be dominated by the diagonal payoff in the same row. Then the sum of the diagonal entries in $C$ must strictly dominate the sum of the off-diagonal entries, which is impossible in a zero-sum game.

\subsection*{Structural properties of data sets}

The previous example highlights the interaction of the structure of data sets with the rank of rationalizing games.  The results in this paper focus on exactly this interaction.

In particular, our results depend on the intersections between subgames in the given data set. Formally, we say that two subgames $(I,J)$ and $(I',J')$ \emph{cross} if $(I \times J) \cap (I' \times J') \neq \emptyset$, but $(I \times J) \not \subseteq (I' \times J')$ and $(I' \times J') \not \subseteq (I \times J)$.  It turns out that when data sets do not contain such crossings, they are in some sense ``simple.'' which motivates the consideration of \emph{laminarity}.

\begin{definition}
\emph{A data set $T$ is \textbf{laminar} if no two subgames in $T$ cross.}
\end{definition}

For more complex data sets that are not laminar, we define the
\emph{crossing set} as the set of all subgames in $T$ that cross some subgame
in $T$. Our next definition is useful in obtaining bounds on the rank
of game that rationalizes a non-laminar data set. Let $\osp_\cross$ be the set of observed choices for subgames in the crossing set.

\begin{definition}
\label{defn:crossing-span}
\emph{For a data set, the \textbf{crossing span} is the minimum of the number of rows and columns spanned by $\osp_\cross$: $\min \left\{  \left|\{ i \mid (i,j) \in \osp_\cross \} \right| \textrm{,  }   \left|\{ j \mid (i,j) \in \osp_\cross \} \right| \right\}.$}
\end{definition}

The crossing span is a natural formalization of the ``richness'' of a data set.  To illustrate the definitions above, Figure~\ref{fig:laminarity} shows a laminar family of subgames, a family where two subgames cross, and a data set (including observed choices) where the crossing span is one.

\begin{figure}[t]
\centering
\subfloat[A laminar family]{ \hspace{.2in}
\includegraphics[scale=0.20]{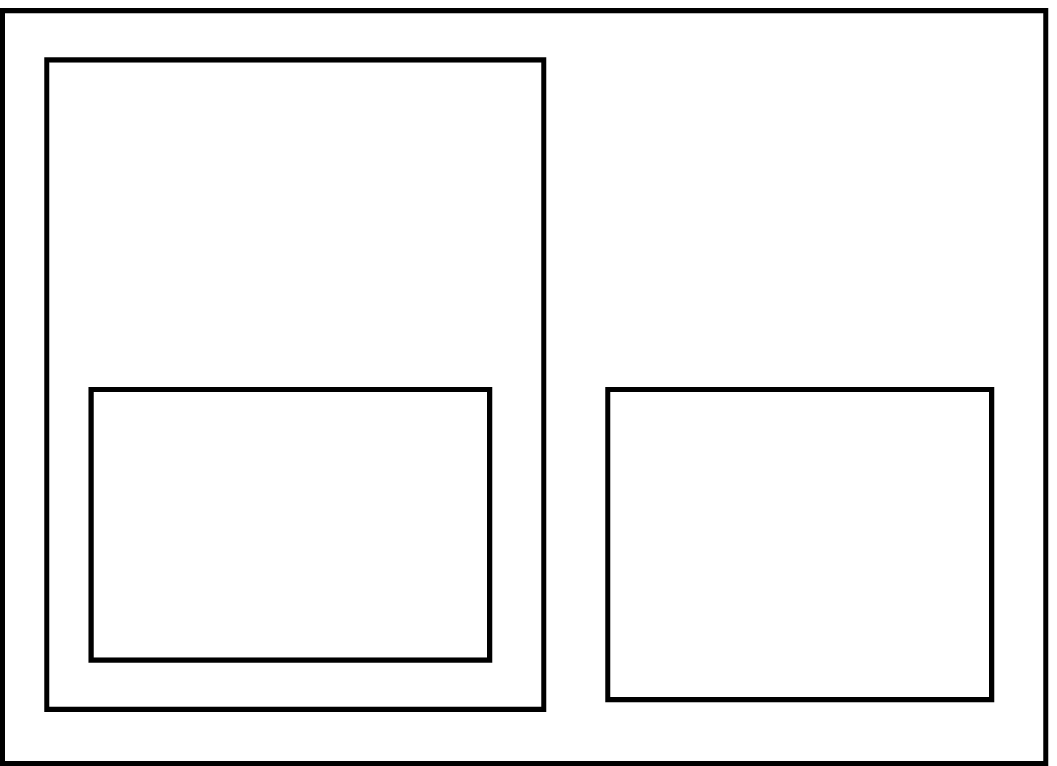} \hspace{.2in}
\label{fig:laminar}}
\qquad
\subfloat[A family where two subgames cross]{  \hspace{.2in}
\includegraphics[scale=0.20]{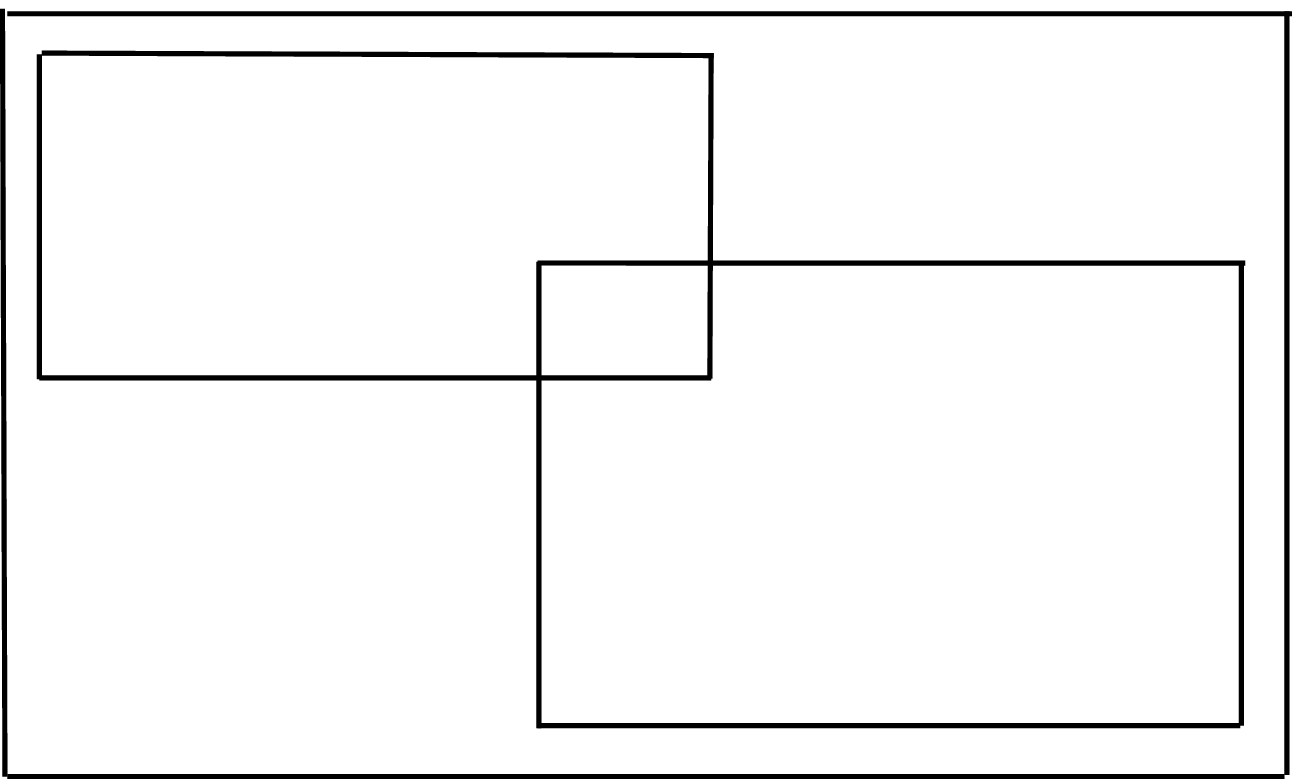}  \hspace{.2in}
\label{fig:intersect}}
\qquad
\subfloat[A family with crossing spans of the row player, column player, and game of $2$, $1$, and $1$]{
\hspace{.5in} \includegraphics[scale=0.20]{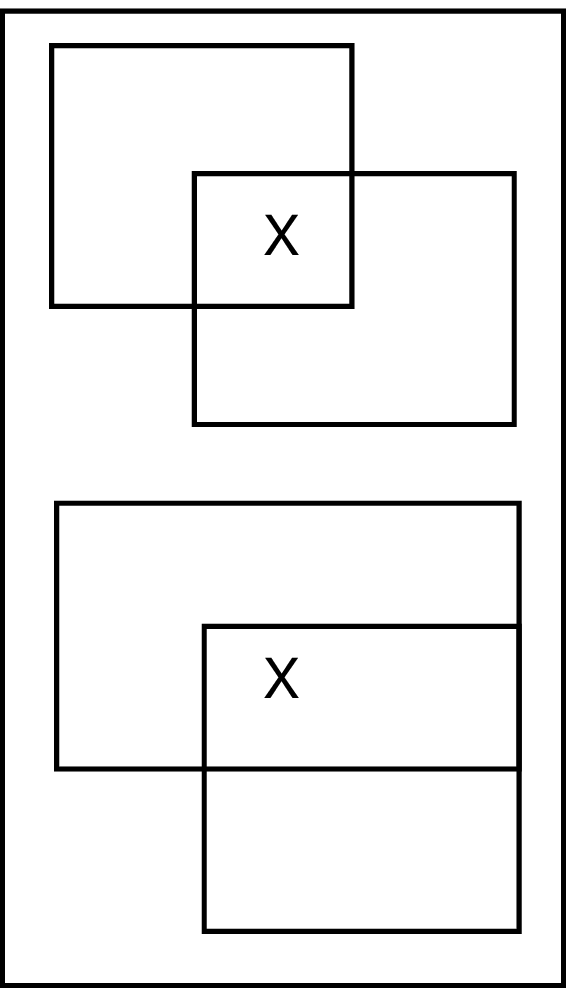} \hspace{.5in}
\label{fig:intersectionspan}}
\caption{\emph{Laminar and crossing families of data sets}}
\label{fig:laminarity}
\end{figure}

In addition to the structural complexity of the data set, the existence of multiple equilibria in a single subgame can influence rank. In fact, there exist (see Section~\ref{sec:high-rank}) data sets that are structurally simple (consist of laminar subgames) but can only be rationalized by high-rank games. Therefore, structural simplicity by itself does not guarantee low-rank rationalizations. The additional constraint on data sets that suffices to establish existence of low-rank games is the following: each subgame in the data set contains a unique equilibrium. Formally, we define uniqueness as follows:

\begin{definition}
\label{defn:unique}
\emph{A rationalizable data set $T$ satisfies the \textbf{uniqueness property} if for any subgame $(I,J)$ in $T$ there is exactly one observed choice $(i,j)$. Further, if $((i,j),I,J), ((i',j'), I', J') \in T $ such that $I' \times J' \subseteq I \times J$ and $(i,j) \in I' \times J'$ then $(i',j') = (i,j)$.}
\end{definition}

\section{Results and Discussion}

The goal of this paper is to test the empirical implications of rank in bimatrix games. 
Toward that end, our first result highlights that the assumption of low rank has empirical bite.  More specifically, the assumption of low rank is refutable via data, i.e., there exist data sets that are rationalizable only via high-rank games.

\begin{theorem}
For all $n$, there exists a rationalizable data set $T$ over an $n \times n$ strategy space such that the rank of any bimatrix game that rationalizes $T$ is $\Omega(\sqrt{n})$.
\label{thm:highrank}
\label{THM:HIGHRANK}
\end{theorem}

Theorem \ref{thm:highrank} highlights that complexity in the structure
of the observed data can manifest in requirements on the rationalizing
game.  This is in contrast to the case of consumer choice theory,
where \cite{echengolovinwierman2011} has shown that whenever observed
choice behavior is
rationalizable, it is rationalizable via an easily-optimizable utility function.
The proof of Theorem~\ref{thm:highrank} in
Section~\ref{sec:high-rank} exhibits a family of data sets that require high
rank rationalizations.  The family of data sets used is
constructed from Hadamard matrices (introduced in~\cite{sylvester1867}), and the $\Omega(\sqrt{n})$ bound is a consequence of the best known lower bound on the min-rank of Hadamard matrices, $\sqrt{n}$~\cite{Hogben11}.

Theorem \ref{thm:highrank} highlights that there exist data sets with structure that require high-rank rationalizing games, and thus motivates the study of data sets with more structure in order to understand when rationalizing games with low rank are possible. In the following we present theorems characterizing classes of data sets that allow low-rank rationalizations. Though simple, the classes of data we consider are natural and relevant for typical experimental designs in economics.

First, we focus on the case when the observed data does not involve subgames, i.e., the observed data consists of repeated observations of choices given fixed strategic options for the players.  In this case, we prove that a rank-one rationalization is always possible, regardless of the number of observations in the data. As shown by Figure \ref{fig:rat}, zero-sum rationalizations are not always feasible for such data sets, and so this result is tight.

\begin{theorem}
For every rationalizable data set of the form $T=\{ ((i,j), [n], [n] ) \}_{i,j}$ there exists a rationalizing  rank-one bimatrix game.
\label{thm:rankone}
\label{THm:RANKONE}
\end{theorem}

In contrast to Theorem~\ref{thm:highrank}, Theorem~\ref{thm:rankone} shows that the assumption of low rank does not have empirical bite for this class of data sets. 
So, the data is not ``rich'' enough to reject low-rank without also rejecting rationality.  Thus, for this class of data, the message is similar to that of \cite{echengolovinwierman2011} for consumer choice theory.  The proof of Theorem~\ref{thm:rankone}, given in Section~\ref{sec:entire-game}, proceeds by constructing explicit payoffs for the players that are maximized at the given observations. The observations thus correspond to strict equilibria in the game. The sum of the payoffs can be described by the outer product of two vectors, and hence the game is of rank one.  Importantly, the construction ensures that the subgame equilibria in the rank-one rationalization are \emph{exactly} the observations in the data set, so the simplicity of the game is not a result of adding equilibria that were not observed.

Our remaining results move away from assumption of fixed strategic choices for the players, and again consider data coming from observations of choices when playing subgames.  In this context, we make two assumptions about the data sets that serve to simplify their structure and guarantee low-rank rationalizations are possible: (i) we assume that the data set satisfies the uniqueness property; and (ii) the observed subgames are laminar.  These two assumptions are motivated by the construction of the proof of Theorem \ref{thm:highrank}, which highlights that high-rank rationalizations require cyclic ``revealed preferences'' that can result from overlapping subgames with multiple observations.  Given these two assumptions about the data set, we prove that the data can be rationalized as a zero-sum game.

\begin{theorem}
If a data set is laminar and satisfies the uniqueness property then it can be rationalized by a zero-sum game. 
\label{thm:laminar}
\label{THM:LAMINAR}
\end{theorem}

Like the case of Theorem \ref{thm:rankone}, Theorem \ref{thm:laminar} states that the assumption of low rank has no empirical bite if the data has a simple structure, i.e., laminar subgames with one observation per subgame.   Zero-sum games are arguably the most intuitive and simple noncooperative games, and it is remarkable that laminar data sets with the uniqueness property, which can conceivably possess a rich structure, can be rationalized by zero-sum games. Our result hence shows that the NP-hardness in the result of~\cite{galambos} can be traced to non-laminar sets of alternatives. The proof of Theorem~\ref{thm:laminar} in Section~\ref{sec:laminar} is based on \emph{revealed-preference graphs}: graphs which
capture the inequalities that must exist between payoffs for the
observed data to be rationalizable. These graphs are a natural component of revealed-preference analysis. In our case,
acyclicity of the revealed-preference graph exactly
corresponds to rationalizability by zero-sum games. We prove that the
revealed-preference graph constructed from laminar data sets with one
observation per subgame are necessarily acyclic. Hence the data
set is rationalizable by a zero-sum game.  As in the case of Theorem \ref{thm:rankone}, structural simplicity is not a result of the rationalization having subgame equilibria that were not observed.

The identification of laminarity as a structural property ensures that rationalization via zero-sum games motivates the study of data sets that are ``nearly laminar''.  One would hope that ``near laminarity'' would lead to rationalization via low-rank games, and this is indeed the case.  In particular, the following theorem highlights that a data set is rationalizable via a game having low rank when the crossing span of the subgames in the data set is small.

\begin{theorem}
Any rationalizable data set $T$ that satisfies the uniqueness property
can be rationalized by a bimatrix game of rank at most the crossing
span of $T$.
\label{thm:non-laminar}
\end{theorem}

Theorem \ref{thm:non-laminar} provides insight into the connection
between the ``richness'' of the observed data (formalized via the
crossing span) and the ``structural complexity'' of rationalizing
games (formalized via the rank).   In particular, it highlights that
the empirical bite of the assumption of low rank is determined by the
crossing span of the data.  That is, data can be used to refute the
assumption of low rank when the number of observed subgames that cross
each other is high enough. The example we construct using Hadamard
matrices for the proof of Theorem~\ref{thm:highrank} has crossing span
$\Theta(n)$, while in the laminar case of Theorem~\ref{thm:laminar},
the crossing span is zero\footnote{Note that, Theorem~\ref{thm:laminar} is stronger than Theorem~\ref{thm:non-laminar} for laminar data sets; since it guarantees rationalizability by a game in which the observed choices are the only equilibria.}.  Thus, Theorem \ref{thm:non-laminar} yields
an $O(n)$ upper bound on the rank of the examples used in the proof of
Theorem \ref{thm:highrank}.  The gap between this upper bound and the
$\Omega(\sqrt{n})$ lower bound in Theorem \ref{thm:highrank}
corresponds to the gap between the best known bounds on the rank of
Hadamard matrices \cite{Hogben11}.

Theorem~\ref{thm:non-laminar} is proven in
Section~\ref{sec:non-laminar}. The proof proceeds through
revealed-preference graphs, similar the proof of
Theorem~\ref{thm:laminar}. However the intersections in the data set
now introduce cycles in the graph. We break the cycles by introducing
additional vertices in the graph. The number of
additional vertices is an upper bound on the rank required for
rationalizing the data set, and is closely related to the crossing
span.  Note that, unlike Theorems \ref{thm:rankone} and \ref{thm:laminar}, the construction does not  ensure uniqueness of equilibria in the subgames.  This is because, even if the data set satisfies uniqueness, if it is not laminar there may be no rationalizing game with unique equilibria in every subgame (see Figure~\ref{fig:highreplace}).

 % Notation, Preliminaries, and Results
\section{Proof of Theorem~\ref{thm:highrank}}
\label{sec:high-rank}

Our proof of Theorem~\ref{thm:highrank} proceeds by explicitly constructing data sets that are rationalizable only by games of rank $\Omega(\sqrt{n})$. In particular, we construct two such data sets. These two data sets are structurally similar, i.e., they enforce the same set of constraints on the payoffs of the players; however, the data sets have different structural properties. One of these data sets is laminar and includes two observations for every subgame in the data set.  The other data set has many intersections among the observations, but satisfies the uniqueness property.  Thus, we prove that both non-laminarity and non-uniqueness of equilibria are sufficient individually to enforce high rank of the rationalizing game.

Throughout, we assume without loss of generality that $n$ is a power of 2. We start by constructing a laminar data set where any rationalizing game has rank $\Omega(\sqrt{n})$, proving Theorem~\ref{thm:highrank}. Then, we show how the same lower bound can be obtained by a data set that satisfies the uniqueness property but has intersections between subgames.

\subsection{High Rank in Laminar Data Sets}
\label{sec:highrank}

We start by  defining a particular class of data sets used in our construction.

\begin{definition} \emph{A data set $T$ on an $n \times n$ game is \textbf{2-regular} if:
\begin{enumerate}
\item Subgame $(A',B')$ appears in $T$ iff $A' = \{2i-1,2i\}$ and $B' = \{2j-1,2j\}$ for $i,j \in [n/2]$;
\item For each subgame $(\{2i-1,2i\},\{2j-1,2j\})$, either $((2i-1,2j-1),\{2i-1,2i\},\{2j-1,2j\} ) \in T$ and $((2i,2j),\{2i-1,2i\},\{2j-1,2j\}) \in T$, or $((2i-1,2j),\{2i-1,2i\},\{2j-1,2j\}) \in T$ and $((2i,2j-1),\{2i-1,2i\},\{2j-1,2j\}) \in T$, but not both.
\end{enumerate}
}
\end{definition}

Two-regular data sets have a very intuitive graphical representation, and one such data set is illustrated in Figure \ref{fig:2regular}. Each $2\times2$ subgame with contiguous rows and columns that starts on even row and column indices appears exactly twice in the data set. The two entries for each subgame correspond to either the diagonal elements of the subgame, or the off-diagonal elements. 

\begin{figure}[t]
\centering
\includegraphics[scale=.53]{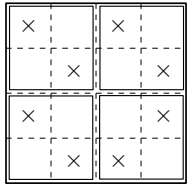}
\caption{\emph{Example of a two-regular data set.}} \label{fig:2regular}
\end{figure}

Next, we associate a sign-pattern matrix with a 2-regular data set. If the data set is in a game of size $n \times n$, the associated sign-pattern matrix is of size $n/2 \times n/2$, and is obtained by replacing each subgame by +1 or -1 depending on whether the entries in the data set for the subgame correspond to diagonal or off-diagonal entries. Thus, for the subgame $(\{2i-1,2i\},\{2j-1,2j\})$,

\begin{enumerate}
\item If $(2i-1,2j-1)$ and $(2i,2j)$ are the observed choices, replace the subgame by +1.
\item If $(2i,2j-1)$ and $(2i-1,2j)$ are the observed choices, replace the subgame by -1.
\end{enumerate}

\noindent For example, the sign-pattern matrix for the example game in Figure \ref{fig:2regular} is 

\[
\left[ \begin{array}{rr}
				+1 & +1 \\
				+1 & -1
				\end{array} \right] \,  .
\]

The key idea behind our analysis is a  correspondence between the sign-pattern matrices described above and \emph{Hadamard matrices} \cite{Hogben11,sylvester1867}, which are square matrices in which every entry is either $+1$ or $-1$ and whose rows are mutually orthogonal.  For example, a Hadamard matrix of order $2$ is

\[
\bfH_2   = \left[ \begin{array}{rr}
				+1 & +1 \\
				+1 & -1
				\end{array} \right] \, .
\]

\noindent The sign-pattern matrix for the 2-regular bimatrix game in Figure \ref{fig:2regular} is exactly $\bfH_2$. We focus on Hadamard matrices with order $2^k$, $k \in \mathbb{Z}^+$ since we assume $n$ is a power of 2. Given a Hadamard matrix of order $2^{k-1}$, a Hadamard matrix of order $2^k$ is 

\[
\bfH_{2^k}   = \left[ \begin{array}{rr}
				\bfH_{2^{k-1}} & \bfH_{2^{k-1}} \\
				\bfH_{2^{k-1}} & -\bfH_{2^{k-1}}
				\end{array} \right] \, .
\]

The \emph{min-rank} of a Hadamard matrix of order $n$ is the minimum rank over all  matrices $\bfX \in \mathbb{R}^{n \times n}$ with $\sign({\bfX}) = \bfH_n$, where $\sign(A)$ for any matrix is obtained by replacing each entry $a_{ij}$ by $\sign(a_{ij}) \in \{-1,0,1\}$. A crucial result for our analysis is:

\begin{proposition}[\cite{Forster02,Hogben11}]
The min-rank of $\bfH_n$ is at least $\sqrt{n}$.
\label{prop:hminrank}
\end{proposition}

Given the above proposition, the key lemma we use to establish Theorem \ref{thm:highrank} is the following, which shows that for 2-regular data sets where the sign-pattern matrix corresponds to a Hadamard matrix, the rank of any game that rationalizes the data set is $\Omega(\sqrt{n})$.

\begin{lemma}
Let $T$ be a 2-regular data set for which the sign-pattern matrix is $\bfH_n$. Let $(A,B)$ be any game that rationalizes $T$, and $C = A+B$. Then $\rank(\bfC) \ge \sqrt{n}$.
\label{lemma:highrank}
\end{lemma}

From Lemma~\ref{lemma:highrank} it is straightforward to complete the proof of Theorem \ref{thm:highrank}. In particular, we can select observations as diagonal or off-diagonal entries in all the subgames of a $2$-regular data set $T$ such that the sign pattern matches $\bfH_n$.  Then the rank of any game that rationalizes $T$ is $\Omega(\sqrt{n})$.

\begin{proof}[Proof of Lemma \ref{lemma:highrank}]
Define matrix $\bfP_n$ of size $n/2 \times n$ as follows: the $i$th row consists of all zeroes, except for +1 in position $2i-1$ and -1 in position $2i$. Thus,

\[
\bfP_4   = \left[ \begin{array}{rrrr}
				+1 & -1 & 0 & 0 \\
				0 & 0 & +1 & -1
				\end{array} \right] \, .
\]

\noindent It is known that for any matrices ${\bfX}$, ${\bfY}$, $\rank({\bfX}{\bfY}) \le \min \{\rank({\bfX}), \rank({\bfY})\}$. We will show that $\rank(\bfP_n \bfC \bfP^T_n) \ge \sqrt{n}$. The lemma then follows immediately.

The key step in the argument is the following claim.

\begin{claim}
For a $2 \times 2$ bimatrix game $(\bfA', \bfB')$, if the equilibria are exactly the diagonal elements, then $\bfP_2 \bfC' \bfP^T_2 > 0$. Conversely, if the equilibria are exactly the off-diagonal elements, then $\bfP_2 \bfC' \bfP^T_2 < 0$.
\label{clm:p2}
\end{claim}

To prove the claim, let us consider the first part first. Since the diagonal elements are equilibria, $a'_{11} > a'_{21}$ and $b'_{11}  > b'_{12}$, along with, $a'_{22} > a'_{12}$ and $b'_{22} > b'_{21}$.

Adding up these inequalities yields $a'_{11} + b'_{11} + a'_{22} + b'_{22} > a'_{21} + b'_{12} + a'_{12} + b'_{21}$, or $c'_{11} + c'_{22} > c'_{12} + c'_{21}$. Since $\bfP_2 \bfC' \bfP^T_2 = c'_{11} + c'_{22} - c'_{12} - c'_{21}$, the first part of the claim follows. The second part of the claim is easily seen by reversing each of the previous inequalities.

Proceeding with the proof of Lemma \ref{lemma:highrank}, let $\bfL := \bfP_n \bfC \bfP^T_n$. Then $\bfL$ has size $n/2 \times n/2$, and $\bfL_{ij} = C_{2i-1,2i-1} + C_{2j,2j} - C_{2i-1,2j} - C_{2i,2j-1}$. Since $(A,B)$ rationalizes the data set, for each $2 \times 2$ subgame that corresponds to subgames in the data set, the equilibria is either on the diagonal elements, or on the off-diagonal elements. By Claim~\ref{clm:p2}, in the first case, $\bfL_{i,j} > 0$, and in the second case, $\bfL_{i,j} < 0$.

Let ${\bfS}$ be the sign-pattern matrix obtained for the given data set; by the statement of the theorem, ${\bfS} = \bfH_n$. Then it follows by the construction of $\bfS$ that ${\bfS}_{ij} > 0 \Rightarrow \bfL_{ij} > 0$, and ${\bfS}_{ij} < 0 \Rightarrow \bfL_{ij} < 0$. Thus, ${\bfS} = \sign(\bfL)$, and hence $\sign(\bfL) = \bfH_n$. Then by Proposition~\ref{prop:hminrank}, $\rank(\bfP_n \bfC \bfP^T_n)$ $= \rank(\bfL)$ $= \rank(\bfH_n)$ $\ge \sqrt{n}$. 
\end{proof}

\subsection{High Rank in Data Sets with Uniqueness}
\label{sec:highrank2}

The previous construction gives a data set requiring high-rank rationalization that is laminar, but does not satisfy uniqueness. We now show how to modify the laminar data set $T$ to obtain a data set $T'$ that satisfies uniqueness, but has many intersecting subgames.

Let $T$ be a 2-regular data set with sign-pattern matrix $H_n$. The modification to obtain $T'$ is based on the following observation. Each pair of observations corresponding to a subgame in the original data set can be replaced by three observations, as shown in Figure~\ref{fig:highreplace}.

\begin{figure}[t]
\centering
\subfloat[The diagonal case.]{\centering \includegraphics[scale=0.3]{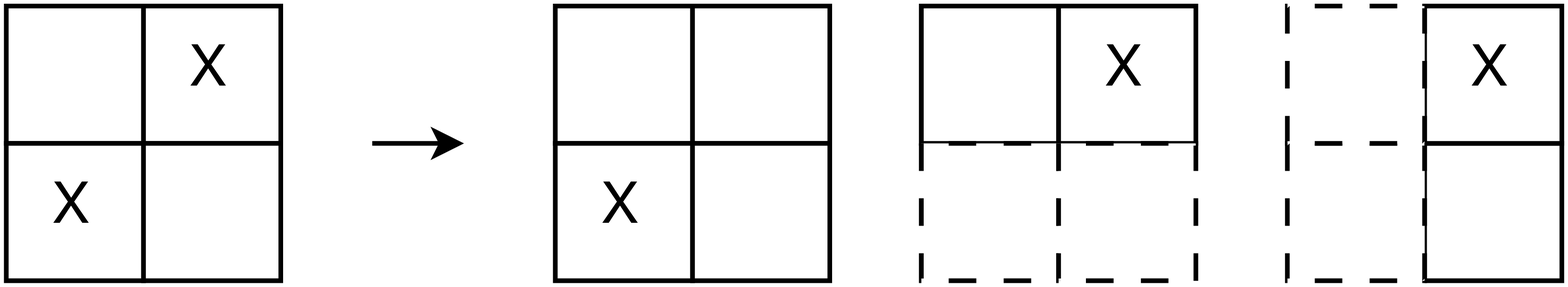}
\label{fig:highrankrep1}}
 \qquad
\subfloat[The off-diagonal case.]{\centering \includegraphics[scale=0.3]{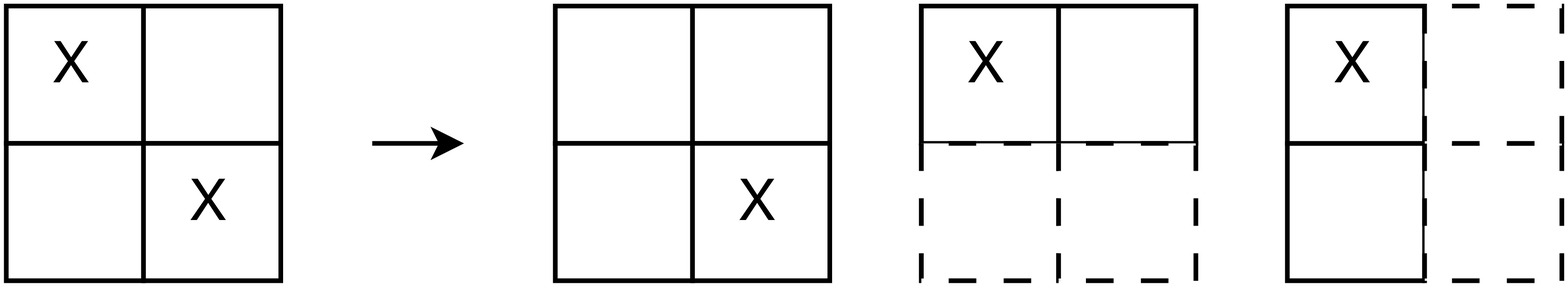}
\label{fig:highrankrep2}}
\caption{\emph{Constructing a data set that satisfies uniqueness. In each of the two cases, the two observations on the left are replaced by the three observations on the right.}}
\label{fig:highreplace}
\end{figure}

This gives us our data set $T'$. Note that we replace two observations
that do not satisfy uniqueness with intersecting subgames. However,
the subgames introduced satisfy the uniqueness property. Thus, the
data set $T'$ is no longer laminar, but does satisfy uniqueness. The
subgames are constructed to capture the same requirements on
payoffs as were induced by the observations with multiplicity of equilibria.

The proof that every rationalization of $T'$ has rank $\Omega(\sqrt{n})$ is very similar to the proof for data set $T$. In particular, we can establish the following claim, corresponding to Claim~\ref{clm:p2}.

\begin{claim}
If the $2 \times 2$ bimatrix game $(\bfA', \bfB')$ rationalizes the observations on the right in Figure~\ref{fig:highrankrep1}, then $\bfP_2 \bfC' \bfP^T_2 > 0$. Conversely, if $(\bfA', \bfB')$ rationalizes the observations on the right in Figure~\ref{fig:highrankrep2}, then $\bfP_2 \bfC' \bfP^T_2 < 0$.
\label{clm:p3}
\end{claim}

The proof of the claim is obtained by observing that the inequalities in Claim~\ref{clm:p2} continue to hold under the modified conditions in Claim~\ref{clm:p3}. The rest of the proof that any rationalizing game for $T'$ has rank $\Omega(\sqrt{n})$, is exactly the same as in the case of $T$.

 %rank lower bound of sort n for general data sets
\section{Proof of Theorem~\ref{thm:rankone}}

\label{sec:entire-game}

Theorem \ref{thm:rankone} considers data of the form $T=\{ ((i,j), [n], [n] ) \}_{i,j}$, which implies that in all the observations the players were free to choose any strategy.  In this section, we prove that whenever data of this form is rationalizable, it is rationalizable with a rank one game. For ease of notation we simply write the observed set $T$ as $\{(i,j)\}_{i,j}$.

To begin the proof, note that if $T$ is rationalizable then, by definition, there exists a bimatrix game, say $(A', B')$, such that every $(i,j) \in T$ is a \emph{strict} Nash equilibrium in $(A',B')$. In other words, for any $(i,j) \in T$, we have $A_{ij} > A_{k j}$ for all $k \neq i$ and $B_{ij} > B_{i k}$ for all $k \neq j$. Therefore all distinct tuples $(i,j)$ and $(i', j')$ in $T$ are in fact component-wise distinct: $i  \neq i'$ \emph{and} $j \neq j'$.

Say we have $\ell$ observations over an $n \times n$ strategy space, $|T| = \ell$; below we construct a $n \times n$ rank-1 bimatrix game $(A,B)$ in which $(i,i)$ is a strict Nash equilibrium for all $i \in [\ell]$ and no other strategy profile is an equilibrium.  Then we can permute the rows of $A$ and $B$ (i.e., relabel strategies) and their columns to get a rank-$1$ game that rationalizes $T$.

The construction is as follows.  Set $A_{ij} = 2ij - i^2 + j^2 $ and $B_{ij} = 2ij + i^2 - j^2$ if $i \in [\ell]$ or $j \in [\ell]$. In addition, set $A_{ij} = 0$ and $B_{ij} = 4ij $ for all $i,j \in \{ \ell+1, \ell + 2, \ldots, n \}$. Note that the $(i, j)$th entry of matrix $A+B$ is $4ij$, hence the game is of rank $1$. 

For all $i \leq \ell$, the largest term in the $i$th column (row) of matrix $A$ ($B$) is on the diagonal. Hence, all the strategy profiles in $\{ (i,i) \mid i \in [\ell] \}$ are strict Nash equilibrium. Note that for all strategy profiles $(i,j) \in \{ \ell+1, \ldots, n \} \times \{ \ell+1, \ldots, n \}$  the row player has a benefiting deviation, implying that none of them can be an equilibrium. In particular, for any such strategy profile we have $A_{1j} > A_{ij}$ (since, $A_{1j} = 2j -1 + j^2$  and $A_{ij}=0$). Thus the set of Nash equilibrium for the game is exactly $\{ (i,i) \mid i \in [\ell] \}$.
 %rank 1 construction when the complete game is observed
\section{Proof of Theorem~\ref{thm:laminar}}
\label{sec:laminar}

In this section we prove that if a data set satisfies laminarity and uniqueness, then it can be rationalized by a zero-sum game. Additionally, every observed subgame in the constructed zero-sum game has a unique equilibrium, corresponding to the observed choices.

Our proof depends on graphs constructed on the payoffs of the row player, called revealed-preference graphs. In a zero-sum game, the payoffs of the row player determine the payoffs of the column player, hence it is sufficient to focus on the payoffs of the row player. Revealed-preference graphs are graphical depictions of the relations that must exist between the payoffs in order for the data set to be rationalizable.

\begin{definition}
\emph{A directed graph $G=(V,E)$ is a \textbf{revealed-preference graph} if there is a bijection $\sigma:V  \rightarrow [n] \times [n]$ so that for every $e=(v,w) \in E$, $\sigma(v)$ and $\sigma(w)$ have exactly one identical coordinate.}
\end{definition}

$G$ is a directed graph, hence edge $e=(v,w)$ is directed from vertex
$v$ to vertex $w$. Correspondence $\sigma$ identifies the vertices of
the revealed-preference graph, hence, with a slight abuse of notation,
we will use $(i,j)$ to denote both strategy profiles and vertices. If
an edge goes between row entries: $e = ((i,j),(i',j))$,  then it is
called a \emph{row edge}. If $e = ((i,j),(i,j'))$, edge $e$ is a
\emph{column} edge. By definition, every edge in a revealed-preference
graph must be either a row edge or a column edge.

\begin{definition}
\emph{A revealed-preference graph $G$ \textbf{implements} an observation $((i,j),X,Y)$ in a data set $T$ if $E$ contains the following edges:
\begin{enumerate}
\item Edges $((i,j),(i',j))$ for each $i' \in X \setminus \{i\}$.
\item Edges $((i,j'),(i,j))$ for each $j' \in Y \setminus \{j\}$.
\end{enumerate}
Further, $G$ \textbf{strongly implements} observation $((i,j),X,Y) \in T$ if $G$ implements the observation and every vertex $(i',j') \neq (i,j) \in X \times Y$ either has a row edge from a vertex in $X \times Y$, or has a column edge to a vertex in $X \times Y$.
}
\label{def:implement}
\end{definition}

We say $G$ (strongly) implements a data set $T$ if it (strongly) implements every observation in $T$. Revealed-preference graphs are useful because of the following implications on the data set they implement.

\begin{lemma}
If revealed-preference graph $G$ implements data set $T$ and is acyclic, then the data set can be rationalized by a zero-sum game.
\label{lem:rpzerosum}
\end{lemma}

\begin{proof}
We explicitly construct the payoff matrices $A,
-A$ for the row and column players and show that these rationalize
the data set. Since $G$ is acyclic, it has a topological ordering. A
topological ordering naturally corresponds to a partial ordering, and
we use this correspondence to choose the entries in $A$. Start with $l = 1$, where $l$ is the length traversed so far in the ordering. For every vertex
$(i,j)$ with no outgoing edge, set $A_{ij} = l$. Remove these vertices
and the incoming edges for these vertices, increment $l$ by 1, and
recurse. It is easy to see that by this construction, if there is a
path from $(i,j)$ to $(i',j')$, then $A_{ij} > A_{i'j'}$. Fill the
remaining entries in $A$ arbitrarily.

We now show that the game $(A,-A)$ rationalizes $T$. Since $G$
implements $T$, for every observation $((i,j),X,Y) \in T$, $G$
contains an edge from $(i,j)$ to each vertex $(i',j)$ for $i' \in X
\setminus \{i\}$. Then $A_{ij} > A_{ij'}$ by construction, and hence
if the column player plays $j$, the best response available to the row
player in the subgame $(X,Y)$ is $i$. Similarly, $G$ contains an edge
from each vertex $(i,j')$ for $j' \in Y \setminus \{j\}$ to
$(i,j)$. Then $A_{ij'} > A_{ij}$, or $-A_{ij} > -A_{ij'}$. Thus if the
row player plays $i$, the best response available to the column player
in the subgame $(X,Y)$ is $j$. Thus $i$ and $j$ are best responses,
and $(i,j)$ must be an equilibrium in the subgame $(X,Y)$. This is
true for every observation, and hence $(A,-A)$ is zero sum and
rationalizes data set $T$.
\end{proof}

\begin{corollary}
If $G$ strongly implements $T$ and is acyclic, then $T$ can be rationalized by a zero-sum game with the additional property that every subgame in $T$ has a unique equilibrium.
\label{cor:rpzerosumunique}
\end{corollary}

\begin{proof}
Since $G$ is acyclic and implements $T$, we can construct $(A,-A)$ that rationalizes $T$. Since $T$ satisfies uniqueness, every subgame $(X,Y)$ in $T$ has a unique observed choice. We will additionally show that for every subgame $(X,Y)$ in $T$ if $(i',j')$ is not this observed choice, then $(i',j')$ is not an equilibrium in the subgame $(X,Y)$. This proves uniqueness.

If $(i',j')$ is not the unique observed choice for subgame $(X,Y)$ that appears in $T$, by definition of strong implementation, the vertex $(i',j')$ either has a row edge from a vertex $(v,j')$ in $X \times Y$, or has column edge to a vertex $(i',w)$ in $X \times Y$. In the first case, by the construction of $A$, $A_{i'j'} < A_{vj'}$ and the row player has an improving deviation to strategy $v$. In the second case, $-A_{i'j'} < -A_{i'w}$ and the column player has an improving deviation to strategy $w$. In either case, $(i',j')$ is not an equilibrium.
\end{proof}

If data set $T$ includes an observation for the entire game $([n],[n])$, Corollary~\ref{cor:rpzerosumunique} holds for this observation as well, and there is a unique equilibrium in the game $([n],[n])$.

If $T$ is laminar and satisfies uniqueness, we will construct a revealed-preference graph that strongly implements $T$ and is acyclic. By Corollary~\ref{cor:rpzerosumunique}, $T$ must then be rationalizable by a zero-sum game, and every subgame observed in $T$ has a unique equilibrium in this game. We assume that no two observations in $T$ have exactly the same observed choice. This is without loss of generality, since if $O_1 = ((i,j),X,Y) \in T$ and $O_2 = ((i,j),X',Y') \in T$ have the same observed choice $(i,j)$, then (by laminarity) one of $X \times Y$, $X' \times Y'$ must contain the other. We remove the observation corresponding to the smaller subgame to obtain data set $T'$. Then any game that rationalizes $T'$ must also rationalize $T$, and if $(i,j)$ is the unique equilibrium in the larger subgame, it is also the unique equilibrium in the smaller subgame.

Any laminar family of sets has a natural representation as a tree\footnote{Technically we obtain a forest. However the distinction is unimportant in this case, and for simplicity we assume we obtain a tree}, where every set corresponds to a vertex. Vertex $v$ corresponding to set $R$ is a child of $w$ corresponding to set $S$ if $R \subset S$, and no set $Q$ in the family satisfies $R \subset Q \subset S$. If two vertices in this tree are not on the same path to the root, the corresponding sets are disjoint. Thus, sets corresponding to leaves in the tree are disjoint, and do not contain any other sets.

Let $\tree$ be the tree obtained as described above for the subgames in $T$. For a subgame $(X,Y)$ corresponding to vertex $v$ in $\tree$, we use $\child(X,Y)$ to denote the subgames corresponding to children of $v$. Our construction of $G=(V,E)$ is inductive on the height of $\tree$. Let $V = [n] \times [n]$, where $n$ is the size of the game for the data set $T$.

In the base case, the height of $\tree$ is 1, and hence the data set
consists of disjoint subgames and a single observation for each
subgame. We describe the construction for a single observation
$((i,j),X,Y)$. Assume without loss of generality that $(i,j)=
(1,1)$. Add edges in $G$ as described in Definition~\ref{def:implement} so that $G$ implements $T$. For strong implementation, we add row edges $((1,j'),(i',j'))$ to $E$ for each $i' \in X \setminus \{ 1 \}$ and $j' \in Y \setminus \{ 1 \}$. Then every vertex in $X \times Y$ that differs from $(1,1)$ in both coordinates has a row edge from a vertex in $X \times Y$. Since $G$ implements $T$, every vertex in $X \times Y$ that differs from $(1,1)$ in a single coordinate already has either a row edge from $(1,1)$ or a column edge to $(1,1)$. The construction is shown for a single observation in a $3 \times 3$ subgame in Figure~\ref{fig:constraint-graph}.

\begin{figure}[t]
\centering
\subfloat[\emph{A revealed-preference graph that strongly implements the observation ((1,1),[3],[3]).}]{ \hspace{.4in}\includegraphics[scale=0.9]{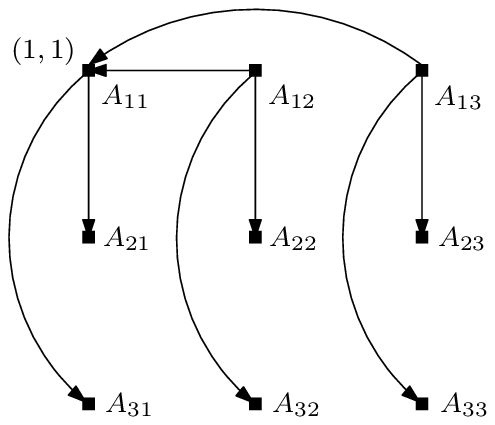}
\label{fig:constraint-graph} \hspace{.3in}}
\qquad
\subfloat[\emph{Revealed-preference graph for subgame $S$.}]{ \hspace{.1in} \includegraphics[scale=.5]{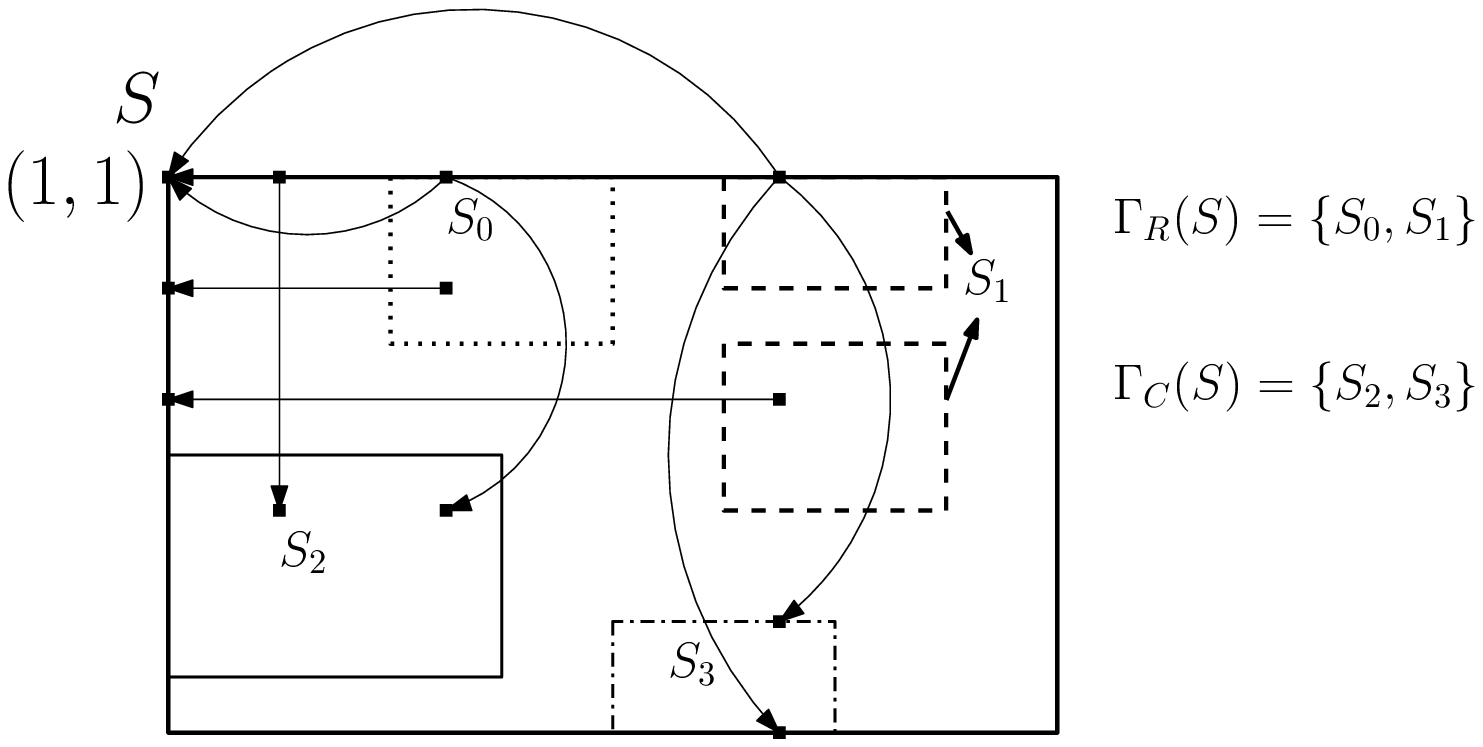} \label{fig:induction} \hspace{.1in}}
\caption{Revealed-preference graphs}
\end{figure}

\begin{claim}
The revealed-preference graph constructed in the base case is acyclic.
\label{clm:baseacyclic}
\end{claim}

\begin{proof}
In the base case, the sets of vertices in V corresponding to each subgame are disjoint. Further, by the construction of the revealed-preference graph, every edge has
both its vertices in the same subgame. Thus all the edges in any cycle must correspond to those added for a single observation. However, for the observation $((i,j),X,Y)$ the only vertex that has both incoming and outgoing edges is $(i,j)$. Hence there cannot be a cycle. 
\end{proof}

For the inductive step, the height of the tree $\tree$ is $k > 1$. We construct $G$ so that it strongly implements all subgames corresponding to vertices in $\tree$ at height less than $k$. By the induction hypothesis, $G$ is also acyclic. For each observation $O=((i,j),X,Y)$ with subgame $(X,Y)$ at height $k$, we add edges to $E$ as follows. As before, we assume without loss of generality that $(i,j) = (1,1)$.

Add edges in $E$ as described in Definition~\ref{def:implement} so that $G$ implements $O$. Let $E1$ and $E2$ denote the set of edges added to satisfy Properties (1) and (2) in Definition~\ref{def:implement} respectively. For strong implementation, we add the edge sets $E3$ and $E4$, described below. In the following discussion, we use $\child_R(X,Y)$ to denote the subgames in $\child(X,Y)$ that contain row one: $\child_R(X,Y) :=\{ (I', J')  \in \child(X,Y) \mid 1 \in I' \}$. We denote the remaining subgames by $\child_C(X,Y) := \child(X,Y) \setminus \child_R(X,Y)$. For any subgame $(X',Y')$, define $V(X',Y') := X' \times Y'$. Define $V_R(X,Y):= \bigcup_{(X',Y') \in \child_R(X,Y)} V(X',Y')$ and $V_C(X,Y) := \bigcup_{(X',Y') \in \child_C(X,Y) } V(X',Y')$. We use $V_L(X,Y)$ to denote the remaining vertices not in the first row or column: $V_L(X,Y) := \{(i,j) \mid i \neq 1, j \neq 1 \mbox{ and } (i,j) \in V(X,Y) \setminus V_R(X,Y) \bigcup V_C(X,Y)\}$.

\begin{itemize}
\item[\emph{E3}:] For every vertex $(i,j) \in V_R(X,Y)$, add edge $((i,j),(i,1))$.
\item[\emph{E4}:] For every vertex $(i,j) \in V_C(X,Y) \bigcup V_L(X,Y)$, add edge $((1,j),(i,j))$.
\end{itemize}

Let $E' = E1 \bigcup E2 \bigcup E3 \bigcup E4$. It is easy to see that the resulting graph strongly implements observation $O=((i,j),X,Y)$. We show now that the resulting graph is also acyclic. The construction and the proof of Claim~\ref{clm:inductacyclic} is also depicted in Figure~\ref{fig:induction}.

\begin{claim}
The graph obtained by the induction step is acyclic.
\label{clm:inductacyclic}
\end{claim}

\begin{proof}
We will show that the added edges $E'$ cannot create any cycles in the revealed-preference graph. Note that by our assumption, no game in $\child(X,Y)$  contains both row $1$ and column $1$, since the observed choice for this subgame would be $(1,1)$ by the uniqueness property of the data set. This implies that for any vertex $(i,j) \in V_R(X,Y)$, $j > 1$, and hence edges in $E3$ always go from a vertex in $V_R(X,Y)$ to a vertex not in $V_R(X,Y)$. Similarly, edges in \emph{E4} go from a vertex not in $V_C(X,Y) \cup V_L(X,Y)$ to a vertex in $V_C(X,Y) \cup V_L(X,Y)$. Thus, no two edges are incident on the same vertex.

We first show that there is no cycle that consists entirely of vertices from a single subgame in $\child(X,Y)$. To see this, note that every edge in $E1$ and $E2$ is incident on $(1,1)$, which is not in any subgame. Further, edges in $E3$ go from a vertex in $V_R(X,Y)$ to a vertex not in $V_R(X,Y)$. Similarly, edges in \emph{E4} go from a vertex not in $V_C(X,Y) \cup V_L(X,Y)$ to a vertex in $V_C(X,Y) \cup V_L(X,Y)$. Thus no edge in $E'$ has both end-points in the same subgame in $\child(X,Y)$, and hence $E'$ cannot create cycles contained in a single subgame in $\child(X,Y)$.

We now show that no cycle can contain any edge that has exactly one incident vertex in a subgame in $\child(X,Y)$. Note that the edges in $E \setminus E'$ introduced earlier have both vertices in the same subgame. Now suppose there is a cycle with such an edge $e=(v,w)$, and the subgame is in $\Gamma_R(X,Y)$. Such an edge must be in $E2$ or $E3$. However, by the way we add these edges, any edge in $E2$ or $E3$ must start from a vertex in $V_R(X,Y)$. There is no way to complete the cycle, since no edge in $E'$ goes to $V_R(X,Y)$. We can similarly show that no edge in a cycle can have exactly one vertex in $\Gamma_C(X,Y)$.

Lastly, it is easy to see that no cycle can only consist of vertices in $V_L(X,Y)$. All edges on these vertices are in $E'$, and by the construction of $E'$, $(1,1)$ is the only such vertex with both an incoming and an outgoing edge.
\end{proof}

The proof of Theorem~\ref{thm:laminar} follows immediately from the inductive construction described above. Specifically, following the construction gives us a revealed-preference graph that strongly implements the data set and is acyclic. Thus, from Corollary~\ref{cor:rpzerosumunique}, we can obtain a zero-sum game $(A,-A)$ that rationalizes $T$ and has a unique equilibrium in every subgame that appears in $T$.

 % rank 0 construction for laminar subgames
\section{Proof of Theorem~\ref{thm:non-laminar}}
\label{sec:non-laminar}

In this section we show that the rank of a rationalizing game is governed by the degree of intersection among subgames in the observed data. In particular, given a rationalizable data set $T$ that satisfies the uniqueness property we can construct a rationalizing bimatrix game of rank at most the crossing span of $T$.

\begin{figure}
\centering
\includegraphics[scale=0.8]{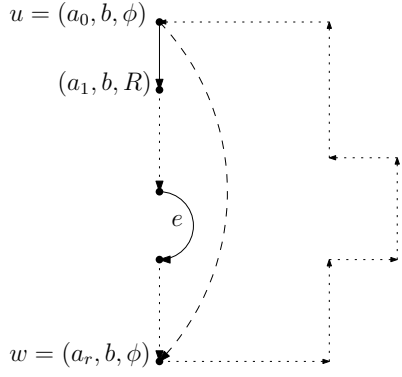}
\label{fig:acyclic}
\caption{\emph{A hypothetical cycle in the revealed-preference graph}}
\end{figure}

For the proof we extend the revealed-preference graphs from Section~\ref{sec:laminar}. If the given data set is not laminar, then the revealed-preference graph can contain cycles. We remove these cycles by replacing vertices corresponding to observed choices in crossing subgames by two vertices, a \emph{row} vertex that is incident only to row edges, and a \emph{column} vertex that is incident only to column edges.

\begin{definition}
\emph{
A directed graph $G=(V, E)$ is a \textbf{split revealed-preference graph} if there is a set $S \subseteq [n] \times [n]$ and a bijection $\sigma: V \rightarrow A \bigcup B$ so that:
\begin{enumerate}
\item Sets $A = \{(i,j,\emptyset)|(i,j) \not \in S\}$, and
  $B = \{(i,j,R),(i,j,C)|(i,j) \in S\}$.
\item For every $e = (v,w) \in E$, exactly one of the following
  statements is true: (a) the first coordinate of  $\sigma(v)$ and
  $\sigma(w)$ are equal, or (b) the second coordinate of  $\sigma(v)$ and
  $\sigma(w)$ are equal. Further, if the
  second coordinate is equal, this edge is a row edge, and the
  third coordinate in both $\sigma(v)$ and $\sigma(w)$ is not $C$. If
  the first coordinate is identical, this edge is a column edge, and
  the third coordinate in both $\sigma(v)$ and $\sigma(w)$ is not $R$.
\end{enumerate}
}
\label{def:split}
\end{definition}

This definition introduces a set $S$ of \emph{split vertices}.
We identify vertices by their
$\sigma$ value; i.e., refer to them as $(i,j,\emptyset)$, $(i,j,R)$,
or $(i,j,C)$. A vertex with $R$ in its third
coordinate is called a row vertex. A
vertex with $C$ in its third coordinate is called a column vertex,
while a vertex with $\emptyset$ as the third coordinate is called an
intact vertex. By the second condition in the definition, no row edges
are incident on column vertices, and no column edges are incident on
row vertices. Thus, a row vertex is never adjacent to a column vertex.

\begin{definition}
\emph{A split revealed-preference graph $G$ \textbf{implements} an observation $((i,j),X,Y)$ in a data set $T$ if $E$ contains the following edges:}
\begin{enumerate}
\item \emph{Edges $((i,j,A),(i',j,B))$ with $A, B \in \{R,\emptyset\}$ for each $i' \in X \setminus \{i\}$.}
\item \emph{Edges $((i,j',A),(i,j,B))$ with $A,B \in \{C, \emptyset\}$ for each $j' \in Y \setminus \{j\}$.}
\end{enumerate}
\label{def:splitimplement}
\end{definition}

As before, $G$ implements data set $T$ if it implements every
observation in the data set. Crucially, if $E$ consists of a minimal set of edges that implement a data set $T$, then the existence of an edge $(v,w)$ where $v = (i,j,A)$ and $w = (i',j',B)$, $A,B \in \{R,C,\emptyset\}$ implies that for some subgame $X \times Y$, both $(i,j) \in X \times Y$ and $(i',j') \in X \times Y$.

For a split revealed-preference graph,
define the \emph{row span} := $|\{i |  (i,j,R) \in V, j \in
[n]\}|$,   \emph{column span} := $|\{j |  (i,j,C) \in V, i \in
[n]\}|$; and \emph{span} of the graph to be the minimum of row and
column span.

\begin{lemma}
If split revealed-preference graph $G$ implements data set $T$ and is acyclic, then $T$ is rationalizable by a game of rank at most the span of $G$.
\end{lemma}

\begin{proof}
Our proof explicitly constructs payoff matrices $A,
B$ for the row and column players and shows that these rationalize
the data set.
In our construction, $A_{ij} + B_{ij} = 0$ if there is an intact vertex $(i,j,\emptyset) \in E$. The only strategy profiles $(i,j)$ where $A_{ij} + B_{ij} \neq 0$ correspond to split vertices. It follows immediately that the number of linearly independent rows in $C = A+B$ is bounded by the row span of $G$, the number of linearly independent columns is bounded by the column span, and hence the rank of $C$ is at most the span of $G$.

Since $G$ is acyclic, it has a topological ordering.
We traverse $V$ following such an ordering, starting with vertices that have no outgoing edges. Start with $l = 1$ where $l$ is the length traversed so far in the
ordering. For every vertex
$v$ with no outgoing edge, if $v$ is an intact vertex $(i,j,\emptyset)$, $A_{ij} = l$ and $B_{ij} = -l$. If $v=(i,j,R)$, set $A_{ij} = l$, else set $B_{ij} = -l$. Remove these vertices
and the incoming edges for these vertices, increment $l$ by 1, and
recurse. Fill in the remaining entries in $A$ arbitrarily and for these entries
set $B_{ij} = -A_{ij}$. If there is an
edge from $(i,j)$ to $(i',j')$, then $A_{ij} > A_{i'j'}$, and $B_{ij} < B_{i'j'}$.

We now show that the game $(A,B)$ rationalizes $T$. Since $G$
implements $T$, for every observation $((i,j),X,Y) \in T$, $G$
contains an edge from $(i,j)$ to each vertex $(i',j)$ for $i' \in X
\setminus \{i\}$. Then $A_{ij} > A_{ij'}$ by construction.
Similarly, $G$ contains an edge from each vertex $(i,j')$ for $j' \in
Y \setminus \{j\}$ to $(i,j)$. Then $B_{ij} > B_{ij'}$.
Thus $(i,j)$ must be an equilibrium in the subgame $(X,Y)$.
This is true for every observation, hence $(A,B)$ rationalizes data set $T$.
\end{proof}

Unlike the previous section, we cannot guarantee uniqueness of
equilibria in the rationalizing game. Even if a
data set satisfies uniqueness, there may be no rationalizing game with
unique equilibria in every subgame. An example of this is shown in
Figure~\ref{fig:highreplace}.

Let $\subgames$ denote the subgames contained in data set $T$.
Let $\subgames_\cross$ be the set of crossing subgames in $T$:
$\subgames_\cross:= \{ S \in \subgames \mid \exists S' \in \subgames
\textrm{ such that  } S \textrm{ and } S' \textrm{ cross}\}$ and let
$\subgames_\laminar:= \subgames \setminus \subgames_\cross$. Subgames
in $\subgames_\laminar$  form a
laminar family. Let $\osp$ be the set of observed choices in $T$.
$\osp_\cross$ is the set of observed choices for subgames in
$\subgames_\cross$, and  $\osp_\laminar := \osp
\setminus \osp_\cross$.

We now define a split revealed-preference graph from $T$. Let
$V$ and $E$ be as in
Definition~\ref{def:splitimplement}, so that  $G=(V,E)$ implements the
data set $T$.  %Let $S =  \osp_\cross$ and $\sigma(i,j)=(i,j,\gamma)$ with $\gamma\in\{\emptyset,R,C\}$ chosen to satisfy $S =  \osp_\cross$ and Definition~\ref{def:split}.  
Define $E_\cross$ as the set of edges
incident on the split vertices, and $E_\laminar = E \setminus E_\cross$.

\begin{claim}
The graph $(V, E_\laminar)$ is acyclic.
\label{clm:nlacyclic1}
\end{claim}

\begin{proof}
Let $T'$ be the set of observations in $T$ over subgames in
$\subgames_\laminar$. The revealed-preference graph defined from $T'$,
$G'=(V',E')$  is acyclic and implements $T'$. Let
$\rho((i,j,\emptyset)) \rightarrow (i,j)$ map
 intact vertices in $V$ to vertices in $V'$. Each edge in $E_\laminar$
  implements an observation in $T'$, since all edges to
 implement observations in $T \setminus T'$ are incident on split
 vertices. As $G'$ implements $T'$, the mapping $\rho$ must
 preserve edges: if there is an edge $(v,w)$ in $E_\laminar$, then
 there is an edge $(\rho(v),\rho(w))$ in $E'$. Since $G'$ is
 acyclic, so is  $(V,E_\laminar)$.
\end{proof}

\begin{claim}
If $T$ is rationalizable, then no cycle in $G$ consists entirely of row edges or entirely of column edges.
\label{clm:nlboth}
\end{claim}

\begin{proof}
Suppose for a contradiction that $T$ is rationalizable and $G$ has
such a cycle $K$. Let $(A,B)$ be a game that rationalizes $T$. Assume
that cycle $K$ is entirely on row edges; the case where $K$ is
entirely on column edges can be ruled out similarly. Every vertex in
$K$ must have the same second coordinate, say $j$. By definition of
revealed-preference graph, there is an edge  $e = ((i,j),(i',j)) \in
K$ only if $A_{ij} < A_{i'j}$. This is obviously incompatible with the
cycle $K$.
\end{proof}

\begin{lemma}Let $T$ be rationalizable. 
The graph $G$ is acyclic and has span exactly equal to the crossing span of the data set.
\end{lemma}

\begin{proof}
We first prove  acyclicity. We will show that if there is a cycle
in $G$, then there is a cycle containing only edges in
$E_\laminar$. This contradicts Claim~\ref{clm:nlacyclic1}, and hence
$G$ must be acyclic. Suppose for the contradiction that $G$ does
contain a cycle $K$. By Claim~\ref{clm:nlacyclic1}, $K$ must contain
an edge $e$ incident on a split vertex $v$.  Assume that $e$ is a row edge; the case where $e$ is a column edge can be handled similarly. Then $v$ must be a row vertex.

By Claim~\ref{clm:nlboth}, $K$ cannot contain only split vertices, as
these would all be on the same row or column.
Let $p'=(v_1, \dots, v_{r-1})$ be a path that contains $v$, is part of
cycle $K$, and the vertex $u$ preceding $v_1$ and $w$
succeding $v_{r-1}$ in $K$ are intact. Observe that $u\neq w$ by
Claim~\ref{clm:nlboth}. 

Let  $p = (v_0 = u, v_1, \dots, v_r =w)$ be the path in $K$ between
$u$ and $w$. We claim that $v_0$ has an 
edge to each $v_i$ in $p$. Then in particular there is an edge $(u,w)$ that is in $E_\laminar$, and we can remove split vertex $v$ from cycle
$K$. Continuing in this manner, we will be left with a cycle only on
intact vertices, which contradicts Claim~\ref{clm:nlacyclic1}. 

Let $v_0 = (a_0, b, \emptyset)$, $v_r = (a_r, b, \emptyset)$, and $v_i
= (a_i,b,R)$ for $1 < i < r$. Our proof is by induction on $i$. First,
the edge $(v_0, v_1)$ is incoming on vertex $v_1$ and is a 
row edge; it must correspond to an observation $O =
((a_0,b),X,Y)$.
By construction of the revealed-preference graph,
the existence of edge $(v_0, v_1)$ implies that $v_1 = (a_1,b) \in X \times Y$. 

Now the inductive step. Let $(a_{i-1},b) \in X \times Y$. 
Edge $(v_{i-1}, v_i)$ is incoming on vertex $v_i$ so it must
correspond to an  observation $O = ((a_{i-1},b),X',Y')$.
By construction of the revealed-preference graph, $v_i = (a_i,b) \in
X' \times Y'$.  Importantly, 
subgame $(X',Y')$ does not cross $(X,Y)$: otherwise the observation
$(X,Y) \in \subgames_\cross$, and $(a_0,b)$ would be a split vertex. Since
$(a_{i-1},b) \in X \times Y \cap X' \times Y'$, either $X \times Y
\subset X' \times Y' $ or $ X' \times Y' \subset X \times Y$. The
former is ruled out by the uniqueness property since otherwise there
are two observations $((a_0,b),X,Y)$ and $((a_{i-1},b),X',Y')$ with $X
\times Y \subset X' \times Y'$ and $(a_{i-1},b) \in X \times Y$, but
$(a_{i-1},b) \neq (a_0,b)$. Thus, $(a_i,b) \in X \times Y$, and
there is a row edge from $(a_0,b)$ to $(a_i,b)$.

For the proof of the span of $G$, fix $(i,j) \in [n] \times [n]$. By
the construction of $G$, there is a split vertex pair
$((i,j,R),(i,j,C)$ if and only if $(i,j) \in \osp_\cross$. It follows
that the span of $G$ must be exactly identical to the crossing span of
$T$. 
\end{proof}

 % low rank construction for non-laminar subgames

\section{Concluding remarks}

This paper characterizes the empirical implications of rank in bimatrix games.  Our results present a fairly complete characterization of how the empirical implications of structural complexity (formalized via rank) depend on the richness of the observed data (formalized via the crossing span).  In particular, we show that the observed data must have a large enough crossing span in order for low-rank to be refutable.

The results in the paper motivate a number of interesting directions.  For example, there is a $\Theta(\sqrt{n})$ gap between Theorems \ref{thm:highrank} and \ref{thm:non-laminar} due, in part, to the fact that our lower bound depends on Hadamard matrices, for which the rank is not precisely known. 

The study of both mixed Nash and correlated equilibria is of great interest, especially given recent results on the hardness of computing mixed Nash. However, it is not trivial to model observations for either of these solution concepts. Treating observed behavior as a sample of theoretical mixed choices is not easy to set up conceptually. In fact the literature in economics has yet to produce results of this case. We believe that our results for pure Nash equilibrium give insight into these solution concepts as well. 

More broadly, the empirical perspective on complexity is under-explored. There exist initial results for consumer choice theory \cite{echengolovinwierman2011}, general equilibrium theory \cite{echengolovinwierman2011}, and the theory of non-cooperative games (this paper).  However, each of these areas warrant further study.  Additionally, the empirical approach can shed light on the testable implications of complexity in other economic theories, such as the theory of stable matchings.

\bibliographystyle{plain}
\bibliography{rv-nash}

\end{document}